\documentclass[12pt,journal,compsoc,onecolumn,]{IEEEtran}

\ifCLASSOPTIONcompsoc
  \usepackage[nocompress]{cite}
\else
  \usepackage{cite}
\fi
\usepackage[total={6.5in, 9in}]{geometry}
\usepackage{setspace}
\usepackage{amsmath}
\usepackage{amsfonts}
\usepackage{amssymb}
\usepackage{stmaryrd}
\usepackage{float}
\usepackage{subcaption}
\usepackage{graphicx}%
\usepackage{algorithm}
\usepackage{algorithmic}

\newtheorem{theorem}{Theorem}

\newtheorem{corollary}[theorem]{Corollary}

\newtheorem{lemma}[theorem]{Lemma}

\newtheorem{proposition}[theorem]{Proposition}
\newtheorem{remark}[theorem]{Remark}

\newenvironment{proof}[1][Proof]{\noindent\textbf{#1.} }{\ \rule{0.5em}{0.5em}}
\begin{document}
\title{Edge-Disjoint Node-Independent Spanning Trees in Dense Gaussian Networks}

\author{Bader~AlBdaiwi,
        Zaid~Hussain,
        Anton~Cerny,
        and~Robert~Aldred
\thanks{B. AlBdaiwi and Z. Hussain are with the Computer Science Department, Kuwait University, Kuwait. P.O. Box 5969 Safat 13060. 
E-mail: \{bdaiwi,zhussain\}@cs.ku.edu.kw}
\thanks{A. Cerny is with the Department of Information Science, Kuwait University, Kuwait.
E-mail: anton.cerny@ku.edu.kw}
\thanks{R. Aldred is with the Department of Mathematics and Statistics, University Of Otago, P.O. Box 56, Dunedin 9054, New Zealand. E-mail: raldred@maths.otago.ac.nz}
}
\maketitle
\begin{abstract}
Independent trees are used in building secure and/or fault-tolerant network communication protocols.
They have been investigated for different network topologies including tori.
Dense Gaussian networks are potential alternatives for 2-dimensional tori.
They have similar topological properties; however,
they are superiors in carrying communications
due to their node-distance distributions and smaller diameters.
In this paper, we present constructions of edge-disjoint node-independent spanning trees in dense Gaussian networks.
Based on the constructed trees, we design algorithms that could be used in fault-tolerant routing or secure message distribution.
\end{abstract}



%

\section{Introduction}\label{sec:introduction}
Parallelism is the performance improvement catalyst in nowadays computing
industry. More computing power and higher scalability are being achieved by
deploying multi-core processing units~\cite{williams2009roofline}. This
could vary from multi-core single processing unit in personal computing
devices to multi-core multi-processing units on servers and supercomputers.
Communications and data exchanges among different processing elements are
usually carried on their interconnection network. Thus, an interconnection
network is a critical limiting factor on a parallel system overall performance
and scalability; consequently, interconnection networks have been receiving considerable research
attention for more than forty years~\cite{dally2004principles}.

The topology of an interconnection network is a main decisive criterion of its
performance qualities. It determines the network communication efficiency as
well as its fault tolerance capabilities. There has been plenty of research on
interconnection topologies like: arrays, trees, hypercubes, butterfly, meshes,
$k$-ary $n$-cubes, and tori \cite{bose1995lee} \cite{dally2004principles}
\cite{leighton2014introduction}. A choice of a topology on which to build a
parallel system is constrained by the target performance objectives,
contemporary available technology, and simplicity. In the 1960's arrays and
two-dimensional meshes were used to build machines like Illiac
\cite{barnes1968illiac} and Soloman \cite{Slotnick:1962:SC:1461518.1461528}.
Later, hybercubes became popular topologies and machines like Cosmic Cube
\cite{seitz1985cosmic} and nCUBES \cite{Ncube:1988:NFH:62297.62415} were built
on them. In late 1980's, meshes, $k$-ary $n$-cube, and toroidal networks
started to become more favorable topologies. Cray T3D \cite{cray1993}, Cray
T3E \cite{scott1996cray}, and Intel Paragon \cite{esser1993intel} are examples
for machines built on these kinds of networks.

An efficient interconnection network called Gaussian network has been
introduced in~\cite{10.1109/TC.2008.57}. It is also called a dense Gaussian network~\cite{martinez2006dense}
when it contains the maximum number of nodes for a given diameter. various facts on some important
versions of such networks, defined in a different equivalent way, have been
described in~\cite{martinez2006dense}. Several related studies on Gaussian
networks followed in~\cite{Flahive:2010:TGE:1850268.1850301}
\cite{Shamaei:2014:HDG:2672598.2673055} \cite{touzene2015all}
\cite{zhang2013efficient}.
Dense Gaussian networks deserve special attention as they are potential alternatives
for $N \times N$ tori.
They have similar structures, same number of nodes, same
number of edges, and both are regular degree four graphs.
Dense Gaussian networks, however,
possesses smaller diameters and better node-distance distributions.
Therefore, they can carry communications more efficiently.

Independent spanning trees can be used in building reliable communication
protocols~\cite{Bao19983}\cite{Chang2015489}\cite{6747363}\cite{DBLP:journals/tc/FragopoulouA96}
\cite{ITAI198843}\cite{Kong:2006:BRM:2692151.2692155}\cite{Touzene:2002:EST:636533.636538}
\cite{DBLP:journals/jpdc/TouzeneDM05}\cite{Tseng96efficientbroadcasting}\cite{DBLP:journals/jpdc/WangB01}
\cite{Yang01012014}\cite{6948321}\cite{10.1007/s11227-014-1346-z}.
Given a network $\mathcal{N}$ in which there exist $t$ independent spanning
trees rooted in $r$, then $r$ can successfully broadcast a message to every
non faulty node $v$ in $\mathcal{N}$ even if up to $t-1$ faulty nodes exist.
Since the number of faulty nodes is less than $t$, then at least one of those
$t$ node disjoint paths from $r$ to $v$ is fault free. Thus, if $r$ broadcasts
through all the existing $t$ independent spanning trees, every non faulty node
in $\mathcal{N}$ would receive the broadcasted message. Similarly, if
$\mathcal{N}$ has $t$ edge disjoint independent spanning trees, all rooted in
$r$, it would be possible to broadcast from $r$ tolerating up to $t-1$ node or
edge faults. The communication applications of independent spanning trees are
not limited to fault tolerance communication. Independent spanning trees could
be utilized in secure message distribution. A message can be split into $t$
packets each of which is sent by $r$ through a different spanning tree to its
destination. In this case only the destination would receive all the $t$
packets while every other node would receive at most one of the $t$ packets
\cite{Lin2010414}\cite{Yang20111254}\cite{Yang:2009:IST:1726593.1728973}. Edge disjoint
independent spanning trees can also be used in building efficient casting
algorithms. A message of size $M$ can be split into $t$ packets, each of size
$M/t$, and simultaneously casted through different trees. Since the trees are
edge disjoint, this splitting could result in less packet collisions and
higher edge utilization, hence it could reduce the average communication time when
wormhole switching is used~\cite{Alsaleh}\cite{dally2004principles}.

In this paper, we present a construction of Edge-Disjoint Node-Independent Spanning Trees
(EDNIST) in dense Gaussian networks.
The rest of the paper is organized as follows. Section~\ref{SectionPrelim}
introduces some necessary preliminaries. Basic facts on Gaussian networks
are briefly summarized in Section~\ref{SectionGaussNet}.
Constructions for two edge-disjoint node-independent spanning trees in
a dense Gaussian network is described in Section~\ref{SectionEDNIST}.
Communication algorithms based on the constructed trees are provided in
Section~\ref{SectionRouting}. Finally, Section~\ref{SectionConcl} concludes
the findings of this paper.

\section{Preliminaries\label{SectionPrelim}}

We will denote by $\mathbb{N}=\{0,1,...\}$ the set of all natural numbers and
by {$\mathbb{Z}=\{  \ldots,-2,-1,0,\allowbreak
1,2,\ldots \}  $ the set of all
integers. An interconnection network is usually modeled as a graph where the
computing elements and their wiring are represented as vertices and edges,
respectively. An undirected \textsl{graph} is a pair $G=(V,E)$ where $V$ is a
set of \textsl{vertices} or \textsl{nodes}, and $E$ is a set of
\textsl{edges }. The edges are unordered pairs. We will use the
notation $V\left(  G\right)  =V$, $E\left(  G\right)  =E$. We are not considering directed graphs in this paper.
If $\left(  u,v\right)  \in E$, we say that
$u$ and $v$ are \textsl{adjacent }and that the vertices $u,v$
and the edge $\left(  u,v\right)$ are \textsl{incident }with each other.
The \textsl{degree} of $v\in V$ is the number of
edges in $E$ with which the vertex $v$ is incident.
The \textsl{union }of two graphs is defined by $\left(  V^{\prime},E^{\prime
}\right)  \cup$ $\left(  V^{\prime\prime},E^{\prime\prime}\right)  =\left(
V^{\prime}\cup V^{\prime\prime},E^{\prime}\cup E^{\prime\prime}\right)  $. A
\textsl{subgraph }of a graph $G$ is any graph $G^{\prime}=(V^{\prime
},E^{\prime})$ such that $V^{\prime}\subseteq V$ and $E^{\prime}\subseteq E$.
Two graphs are \textsl{vertex-disjoint } if their sets
of vertices are disjoint.
Similarly, two graphs are \textsl{edge-disjoint } if their sets
of edges are disjoint. An \textsl{isomorphism} of graphs
$G=\left(  V,E\right)  ,G^{\prime}=\left(  V^{\prime},E^{\prime}\right)  $ is
a bijection $f:V\rightarrow V^{\prime}$ such that $E^{\prime}=\left\{  \left(
f\left(  u\right)  ,f\left(  v\right)  \right)  |\left(  u,v\right)  \in
E\right\}  $; $f$ is an \textsl{automorphism} on $G$ if $G^{\prime}=G$.
A\textsl{ path from vertex }$u$ to vertex $v$ in $G$ is a sequence
$P_{G}\left(  u,v\right)  =\left(  v_{0},v_{1},\ldots,v_{n}\right)  $ of
pairwise distinct (with the possible exception $v_{0}=v_{n}$) vertices from
$V$ such that $u=v_{0}$, $v=v_{n}$, and $\left(  v_{i},v_{i+1}\right)  \in E$
for $i\in\left\{  0,1,\ldots,n-1\right\}  $; $n$ is the \textsl{length} of the
path $P_{G}\left(  u,v\right)  $. We will identify the path $\left(
v_{0},v_{1},\ldots,v_{n}\right)  $ with the subgraph $\left(  \left\{
v_{0},v_{1},\ldots,v_{n}\right\}  ,\left\{  \left(  v_{i},v_{i+1}\right)
|i\in\left\{  0,1,\ldots,n-1\right\}  \right\}  \right)  $. The
\textsl{distance }\ $d\left(  u,v\right)  $ between vertices $u,v$ is the
length of the shortest path from $u$ to $v$. The \textsl{diameter} of the
graph is the maximum distance of two nodes in the graph. A \textsl{cycle }is a
path from $u$ to $u$. A graph is \textsl{connected} if, for any $u,v\in V$,
there is a path from $u$ to $v$ in $G$. A graph is acyclic if it does not contain a cycle of
positive length. A connected acyclic graph is called a \textsl{tree}. In a
tree, there is a unique path between any two vertices. A tree is
\textsl{rooted }if one of its vertices is denoted as its \textsl{root}. The
\textsl{depth }of a rooted tree is the length of the longest path in the tree
starting at the root. Given a graph $G=(V,E)$, a spanning tree of $G$ is a
subgraph $T=\left(  V,E^{\prime}\right)  $ being a tree. It is a well-known
fact that a connected subgraph $\left(  V,E^{\prime}\right)  $ of $G$ is a
spanning tree if and only if $\left\vert E^{\prime}\right\vert =\left\vert
V-1\right\vert $. Spanning trees have applications in network broadcasting in
which a source node sends a message to every other node in the network. A
graph may have a number of different spanning trees. Two rooted spanning trees $T_{1},T_{2}$ of $G=\left(
V,E\right)  $ having the same root $r$ are \textsl{node}
\textsl{independent} if, for any vertex $v\in V$, the only common vertices of
the paths from $r$ to $v$ in $T_{1}$ and $T_{2}$ are $r$ and $v$.

\section{Gaussian Networks\label{SectionGaussNet}}

In this paper, we will deal with a special type of graphs called Gaussian networks. The
name comes from one possible approach to their definition, based on the subset
of complex numbers $\mathbb{Z}\left[  \mathbf{i}\right]  =\left\{
x+y\mathbf{i}|x,y\in\mathbb{Z}\right\}  $ with the usual norm $\left\Vert
x+y\mathbf{i}\right\Vert =x^{2}+y^{2}$. The elements of this set are called
Gaussian integers. It is known that for every $\alpha,\beta\in\mathbb{Z}%
\left[  \mathbf{i}\right]  $, $\alpha\neq0$, there exist unique $q,r\in
Z[\mathbf{i}]$ such that $\beta=q\alpha+r$ with $\left\Vert r\right\Vert
<\left\Vert \alpha\right\Vert $. In analogue with integers, one can write
$\beta\operatorname{mod}\alpha=r$.  A \textsl{Gaussian network}
$G\left(  \alpha\right)  =\left(  \mathbb{Z}\left[  \mathbf{i}\right]
_{\alpha},E\left(  \alpha\right)  \right)  $ is given by a fixed Gaussian
integer $\alpha\neq0$. The set of nodes is the residue class modulo $\alpha$.
The nodes $\beta$ and $\gamma$ are adjacent if and only if $(\beta
-\gamma)\operatorname{mod}\alpha\in\left\{  \pm1,\pm\mathbf{i}\right\}  $.
Various representations of these residue classes are provided in \cite{JordanPotratz}.
Gaussian networks are regular. The degree of each vertex is four. They are
highly symmetric; in fact they are vertex-transitive.
For every pair $u,v\in\mathbb{Z}\left[  \mathbf{i}\right]
_{\alpha},$ there is an automorphism of $G\left(  \alpha\right)  $ mapping
$u$ to $v$. The total number of nodes in the network $G\left(  \alpha\right)
$ is $\left\Vert \alpha\right\Vert $. The distance distribution for Gaussian
networks can be found in~\cite{10.1109/TC.2008.57} where the diameter of the
network $G\left(  x+y\mathbf{i}\right)  $ is equal to $y$ if $\left\Vert
x+y\mathbf{i}\right\Vert $ is even and to $y-1$ otherwise. Gaussian networks are
closely related to circulant graphs. A \textsl{circulant graph} with $N\geq 1$
vertices and two jumps $a,b$, where $1\leq a<b$ is the graph
$C_{N}\left(  a,b\right)$ where $V(C_N(a,b))=\left\{
0,1,\ldots,N-1\right\}  $ and, for $u,v\in V$, $\left(  u,v\right)  \in E(C_N(a,b))$ if
$\left\vert u-v\right\vert \operatorname{mod}N\in\left\{a,b,N-a,N-b\right\}$.
$C_{a^{2}+b^{2}}\left(  a,b\right)  $ is isomorphic to $G(a+bi)$ if and
only if $\gcd\left(  a,b\right)  =1$~\cite{10.1109/TC.2008.57}. A Gaussian
network is \textsl{dense} when it has a maximum number of nodes for a given
diameter $k$. It is shown in \cite{Beivide:1991:ODN:123467.123483} that this
is the case in the graph $C_{k^{2}+\left(  k+1\right)  ^{2}}\left(
k,k+1\right)  $. In this paper, we deal with these dense
diameter-optimal graph, which is isomorphic to the Gaussian network $G\left(  \alpha_{k}\right)  $, where
$\alpha_{k}=k+\left(  k+1\right)  \mathbf{i}$, since $\gcd (k,k+1) = 1$.
We denote $G_{k}=\left(
V_{k},E_{k}\right)  =G\left(  \alpha_{k}\right)  $, $k\geq1$.
We will call an
edge in $G_{k}$ to be \textsl{horizontal} if it is of the form $\left(
u,\left(  u\pm1\right)  \operatorname{mod}\alpha_{k}\right)  $ and
\textsl{vertical} if it is of the form $\left(  u,\left(  u\pm\mathbf{i}\right)
\operatorname{mod}\alpha_{k}\right)  $.
A path in $G_k$ is \textsl{horizontal (vertical)} if does not contain a \textsl{vertical (horizontal)} edge.
We depict the graph $G_{k}$ the
usual way the complex numbers are depicted in the Cartesian plane. The vertex
$a+b\mathbf{i}$, $\left\vert a\right\vert +\left\vert b\right\vert \leq k$,
will be positioned at point $\left\langle a,b\right\rangle $. The graph
$G_{3}$ is depicted in Fig.~\ref{34graph}.

\begin{figure}[H]
\centering
\includegraphics[scale=0.85]{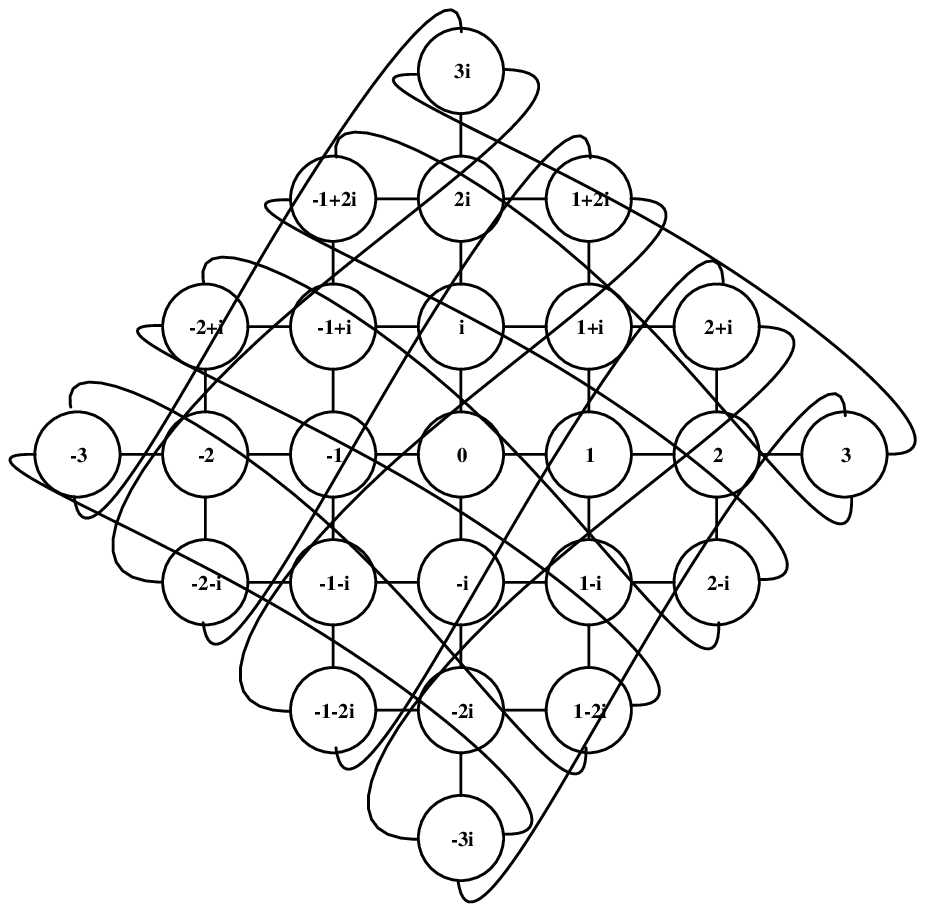} \caption{Gaussian network generated with $\alpha = 3+4i$.}%
\label{34graph}%
\end{figure}

\section{Edge-disjoint node-independent spanning trees in Gaussian
networks\label{SectionEDNIST}}

We will present here a solution to the problem of finding a set of
edge-disjoint node-independent spanning trees in the Gaussian network $G_{k}$.

\bigskip
\begin{proposition}
\label{SpanningSet}Let $\mathcal{T}_{k}$ be a set of edge-disjoint spanning
trees in $G_{k}$, then $\left\vert \mathcal{T}_{k}\right\vert \leq~2$.
\end{proposition}

\begin{proof}
The total number of nodes in $G_{k}=2k^{2}+2k+1$. Hence, each spanning tree in
$\mathcal{T}_{k}$ must have exactly $2k^{2}+2k$ edges. Since the total number
of edges in $G_{k}=4k^{2}+4k+2$, and the trees in $\mathcal{T}_{k}$ are edge
disjoint, it follows that $\left\vert \mathcal{T}_{k}\right\vert \leq2$.
\end{proof}

\bigskip
In our further considerations, we will use automorphisms $\rho$ and $\sigma$
of the Gaussian network $G_{k}$, $k \geq 2$. The mapping $\rho$ is the
$90^{\circ}$ counterclockwise rotation. It is defined, for $a+b\mathbf{i}\in
V_{k}$, by $\rho\left(  a+b\mathbf{i}\right)  =\mathbf{i}\left(
a+b\mathbf{i}\right)  =-b+a\mathbf{i}$. The mapping $\sigma$ is the symmetry
with respect to the axis of the second and the fourth quadrants. It is defined,
for $a+b\mathbf{i}\in V_{k}$, by $\sigma\left(  a+b\mathbf{i}\right)
=-b-a\mathbf{i}$. The mappings satisfy the relations $\rho^{2}\sigma
=\sigma\rho^{2}$ and $\rho^{4}=\sigma^{2}=\mathbf{1}$, where $\mathbf{1}$ is
the identity mapping. Each of $\rho$ and $\sigma$ maps a horizontal edge to a
vertical one and vice versa. We will now describe two subgraphs of $G_{k}$:
the \textquotedblleft black subgraph\textquotedblright\ $B_{k}$, and the
\textquotedblleft red subgraph\textquotedblright\ $R_{k}$, being our
candidates for the two edge-disjoint node-independent spanning trees.
Each of $B_k$ and $R_k$ is the union of several component subgraphs of $G_k$.
These components will be obtained by applying the mappings $\rho$ and $\sigma$ to
the following graphs:
\newline\newline
1. The array subgraph $A$:
\newline$V\left(  A\right)  =\left\{  a+b\mathbf{i|}0\leq a\leq k-1,1\leq
b\leq k,a+b\leq k\right\}  ,$
\newline
$E\left(  A\right)  =\left\{  \left(
a+b\mathbf{i},\left(  a+1\right)  +b\mathbf{i}\right)  |0\leq a\leq k-2,1\leq
b\leq k-1,a+b\leq k-1\right\}  $; $A$ contains all possible horizontal edges
among $V(A)$. \newline The edge size of $A\ $is $\left\vert E\left(  A\right)
\right\vert =k\left(  k-1\right)  /2$.
\newline\newline
2. The baseline subgraph $B$:
\newline$V\left(  B\right)  =\left\{  a\mathbf{|}0\leq a\leq
k\right\}  ,$ \newline$E\left(  B\right)  =\left\{  \left(  a,\left(
a+1\right)  \right)  |0\leq a\leq k-1\right\}  $; $B$ contains all possible
horizontal edges among $V(B)$. \newline The edge size of $B\ $is $\left\vert
E\left(  B\right)  \right\vert =k$.
\newline\newline
3. The $\ B_{k}$-specific wrap-around graph $W^{B}$:\newline$V\left(  W^{B}\right)  =\left\{
a+b\mathbf{i}|,\left\vert a+b\right\vert =k\text{, and }0\leq a,1\leq b\text{
or }a\leq0,b\leq-1\right\}  ,$ \newline$E\left(  W^{B}\right)  =\left\{
\left(  a+b\mathbf{i},-\left(  b-1\right)  +\left(  -1-a\right)
\mathbf{i}\right)  |a+b=k,0\leq a,1\leq b\right\}  $; $W^{B}$ contains all
possible horizontal edges among $V(W^{B})$.\newline The edge size of $W^{B}%
\ $is $\left\vert E\left(  W^{B}\right)  \right\vert =k$.
\newline\newline
4. The $\ R_{k}$-specific wrap-around graph $W^{R}$:\newline$V\left(
W^{R}\right)  =\left\{  a+b\mathbf{i}|1\leq a,b\leq0\text{ or }a\leq
0,b\geq1\text{, and }\left\vert a+b\right\vert =k\right\}  ,$\newline$E\left(
W^{R}\right)  =\left\{  \left(  a+b\mathbf{i},b+a\mathbf{i}\right)
|a-b=k,1\leq a,b\leq0\right\}  $; $W^{R}$ contains all possible horizontal
edges among $V(W^{R})$\newline The edge size of $W^{R}\ \left\vert E\left(
W^{R}\right)  \right\vert =k$.

\bigskip
The subgraph $B_{k}$ is defined as: \begin{figure*}[ht!]
    \centering
    \begin{subfigure}[b]{0.5\textwidth}
        \centering
        \includegraphics[scale=0.45]{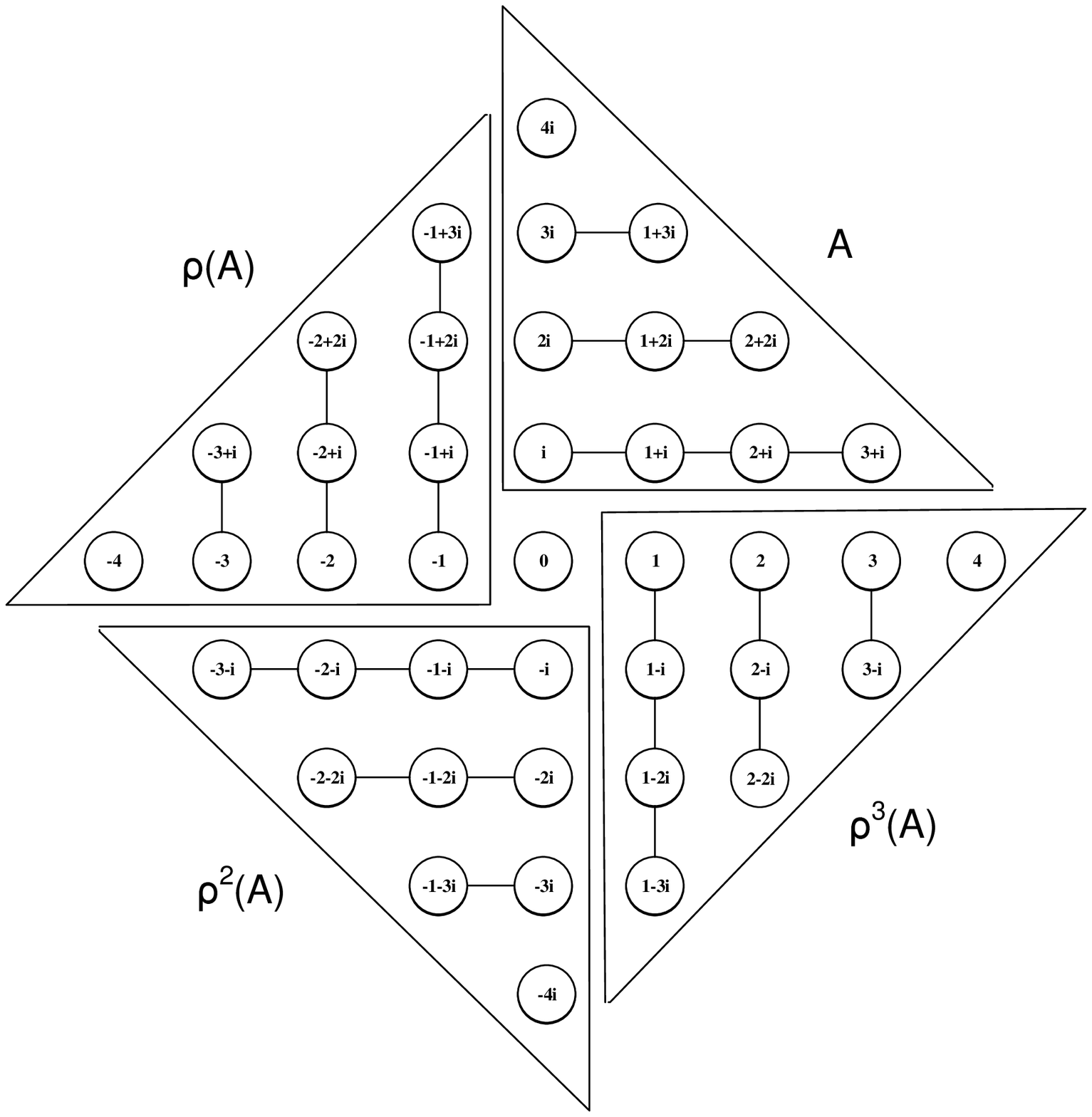}
        \caption{Array graphs}
    \end{subfigure}%
    ~
    \begin{subfigure}[b]{0.5\textwidth}
        \centering
        \includegraphics[scale=0.45]{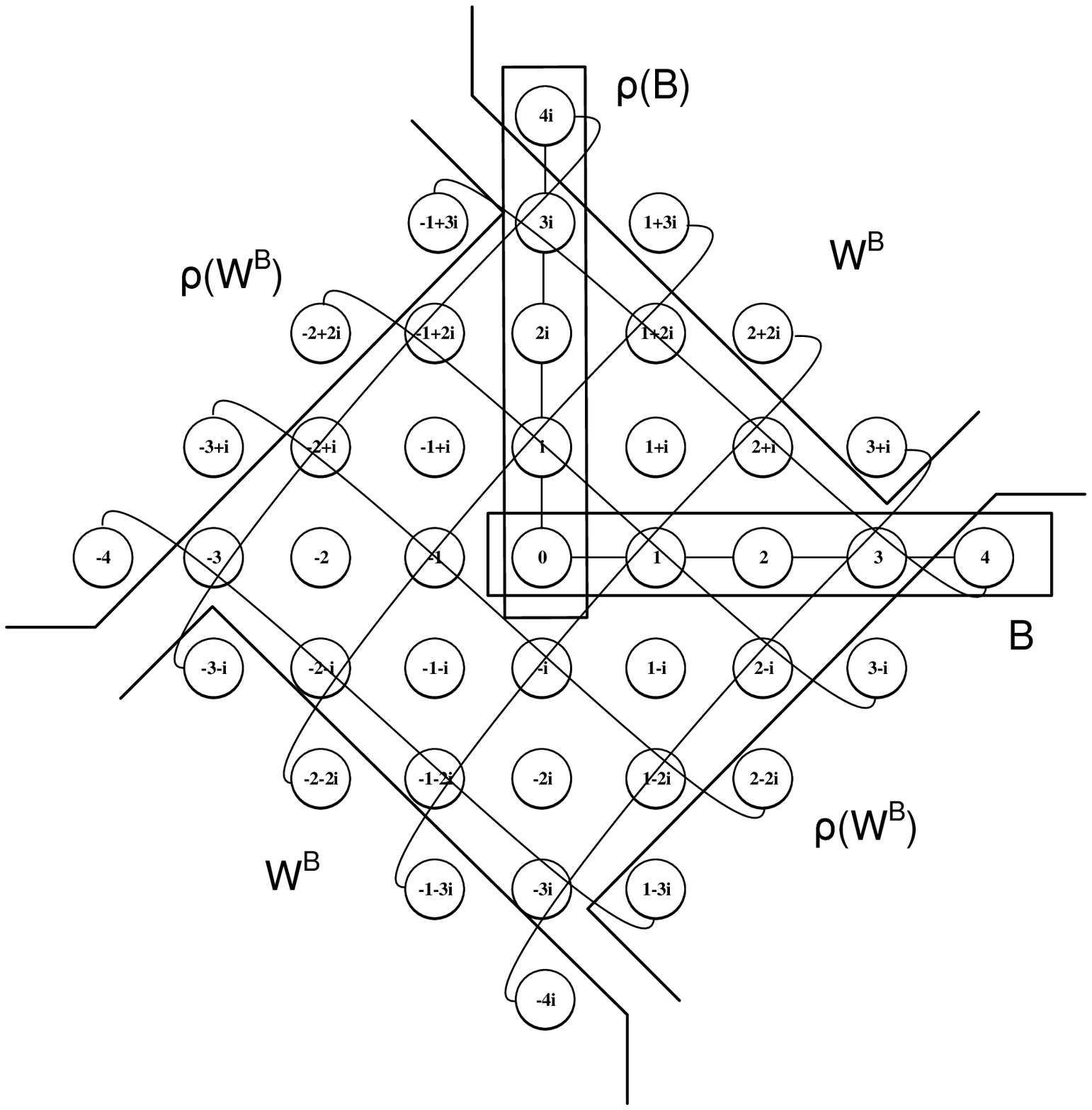}
        \caption{Baseline and $B_4$-specific wrap-around graphs}
    \end{subfigure}
    \caption{$B_4$ components ($k=4$)}
    \label{BlackAll}
\end{figure*}%
\begin{equation}
B_{k}=A\cup\rho\left(  A\right)  \cup\rho^{2}\left(  A\right)  \cup\rho
^{3}\left(  A\right)  \cup B\cup\rho\left(  B\right)  \cup W^{B}\cup
\rho\left(  W^{B}\right)  \label{BU}%
\end{equation}

\bigskip Fig.~\ref{BlackAll} depicts the components of $B_{4}$.

\bigskip The subgraph $R_{k}$ is defined as: \begin{figure*}[t!]
    \centering
    \begin{subfigure}[b]{0.5\textwidth}
        \centering
        \includegraphics[scale=0.45]{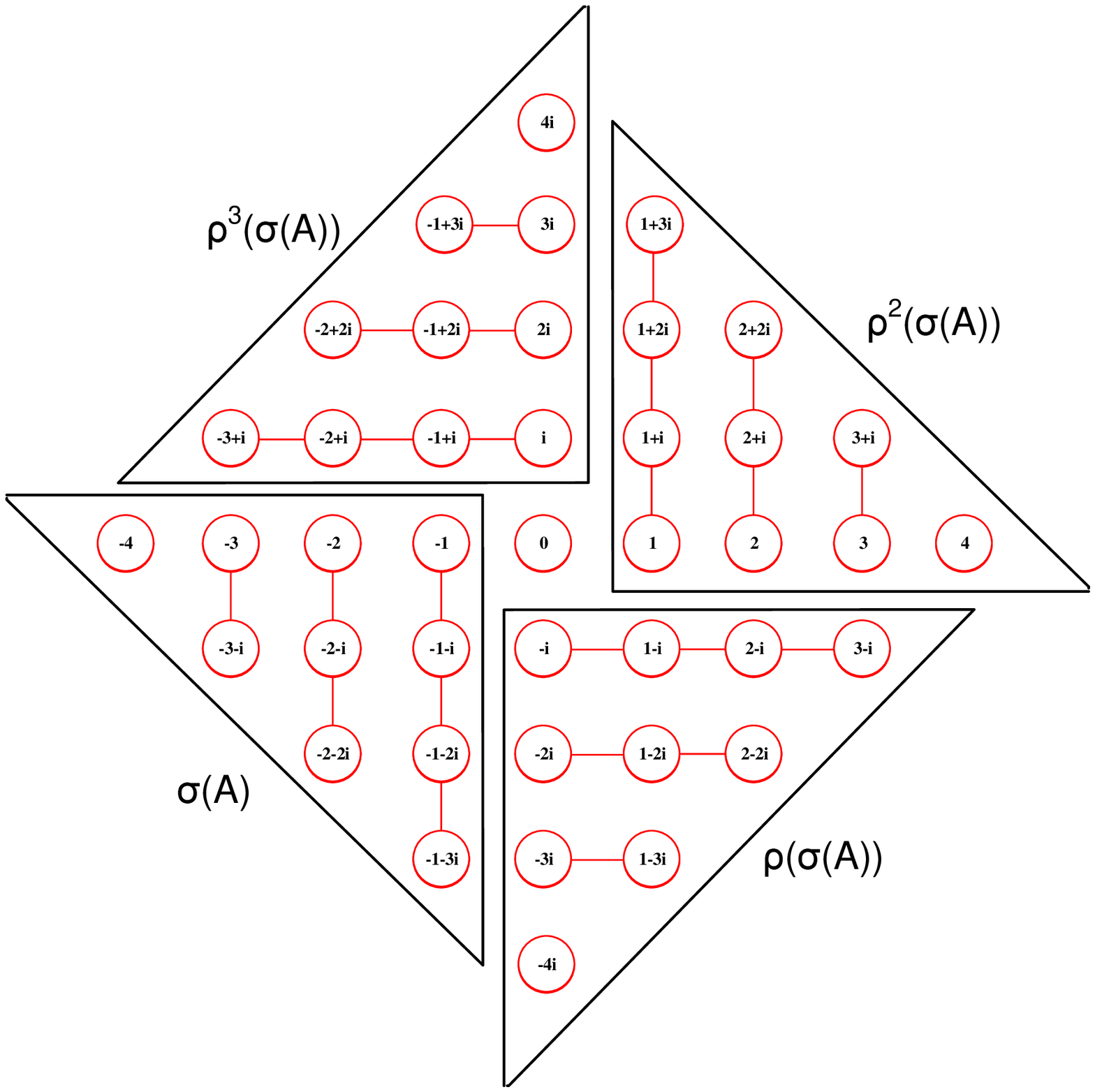}
        \caption{Array graphs}
    \end{subfigure}%
    ~
    \begin{subfigure}[b]{0.5\textwidth}
        \centering
        \includegraphics[scale=0.45]{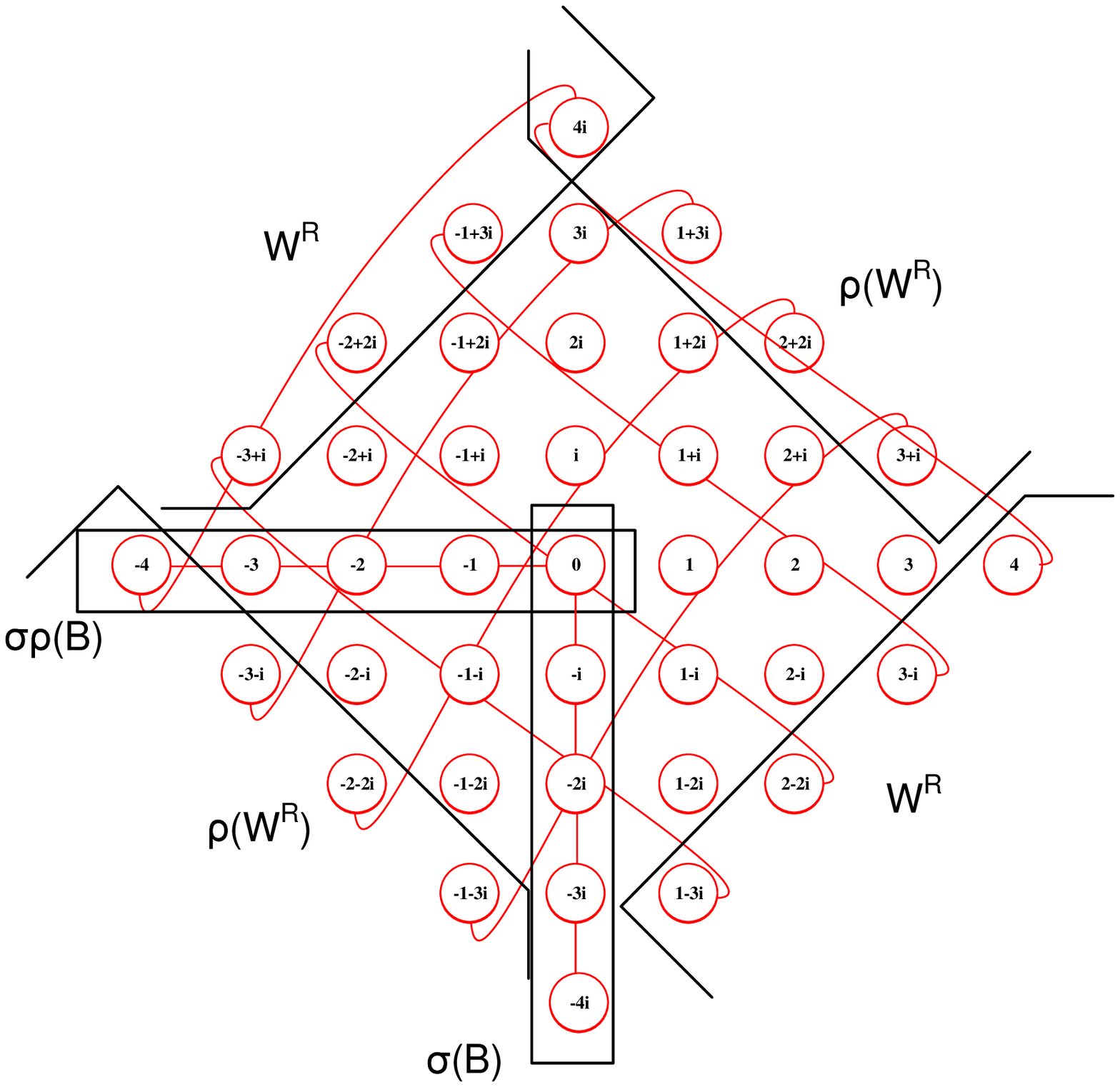}
        \caption{Baseline and $R_4$-specific wrap-around graphs}
    \end{subfigure}
    \caption{$R_4$ components ($k=4$)}
    \label{RedAll}
\end{figure*}%
\begin{equation}
R_{k}=\sigma\left(  A\right)  \cup\rho\sigma\left(  A\right)  \cup\rho
^{2}\sigma\left(  A\right)  \cup\rho^{3}\sigma\left(  A\right)  \cup
\sigma\left(  B\right)  \cup\sigma\rho\left(  B\right)  \cup W^{R}\cup
\rho\left(  W^{R}\right)  \label{RU}%
\end{equation}

\bigskip Fig.~\ref{RedAll} shows the components of $R_{4}$, and Fig.~\ref{45BlackRedST} limns $B_4$ and $R_4$.

\bigskip
\begin{proposition}
\label{BasicRB} $B_{k}$ and $R_{k}$ have the following basic properties:
\newline
1. The vertex set of each of the graphs $B_{k}$ and $R_{k}$ is $V_{k}$.
\newline
2. The edge sets of all components of $B_{k}$, as stated in
$(\ref{BU})$ are pairwise disjoint. The edge sets of all components of $R_{k}$,
as stated in $(\ref{RU})$ are pairwise disjoint.
\newline
3. Each of the graphs $B_{k}$ and $R_{k}$ consists of $2k^{2}+2k$ edges.
\newline
4. The edge sets of $B_{k}$ and $R_{k}$ are disjoint.
\end{proposition}

\begin{proof}
\newline
1. The following can be observed from the construction of the
subgraphs $B_{k}$ and $R_{k}$:%
\begin{align*}
V_{k}  &  =V\left(  A\right)  \cup V\left(  \rho\left(  A\right)  \right)
\cup V\left(  \rho^{2}\left(  A\right)  \right)  \cup V\left(  \rho^{3}\left(
A\right)  \right)  \cup V\left(  B\right) \\
V_{k}  &  =V\left(  \sigma\left(  A\right)  \right)  \cup V\left(  \rho
\sigma\left(  A\right)  \right)  \cup V\left(  \rho^{2}\sigma\left(  A\right)
\right)  \cup V\left(  \rho^{3}\sigma\left(  A\right)  \right)  \cup V\left(
\sigma\left(  B\right)  \right)
\end{align*}
\newline\newline
2. Let $C_{1}, C_{2}$ be two distinct components of $B_{k}$.
If $C_{1},C_{2}$ are not node-disjoint then all edges in one of these
components are horizontal and all edges in the other one are vertical. A
similar argument is valid for $R_{k}$.
\newline\newline
3. Since the components
of $B_{k}$ are edge-disjoint, the edge-size of $B_k$
is the sum of the edge-sizes of the components. The mappings $\rho,\sigma$
preserve the edge-size of a graph. The same argument applies to $R_k$.
Therefore, a summation of the edge-disjoint component sizes yields:%
\begin{align*}
\left\vert E\left(  B_{k}\right)  \right\vert  &  =4\left\vert E\left(
A\right)  \right\vert +2\left\vert E\left(  B\right)  \right\vert +2\left\vert
E\left(  W^{B}\right)  \right\vert \\
&  =4\Big[\frac{k\left(  k-1\right)  }{2}\Big]+2k+2k=2k^{2}+2k\\
\left\vert E\left(  R_{k}\right)  \right\vert  &  =4\left\vert E\left(
A\right)  \right\vert +2\left\vert E\left(  B\right)  \right\vert +2\left\vert
E\left(  W^{R}\right)  \right\vert \\
&  =4\Big[\frac{k\left(  k-1\right)  }{2}\Big]+2k+2k=2k^{2}+2k\text{.}
\end{align*}
\newline\newline
4. Comparing any component of $B_{k}$ to any component of
$R_{k}$, one can observe that such pair of components either does not contain a
common pair of vertices or the direction of the edges in the two components
are different. Hence, $B_{k}$ and $R_{k}$ are edge-disjoint.
\end{proof}

\begin{figure}[H]
\centering
\includegraphics[scale=0.95]{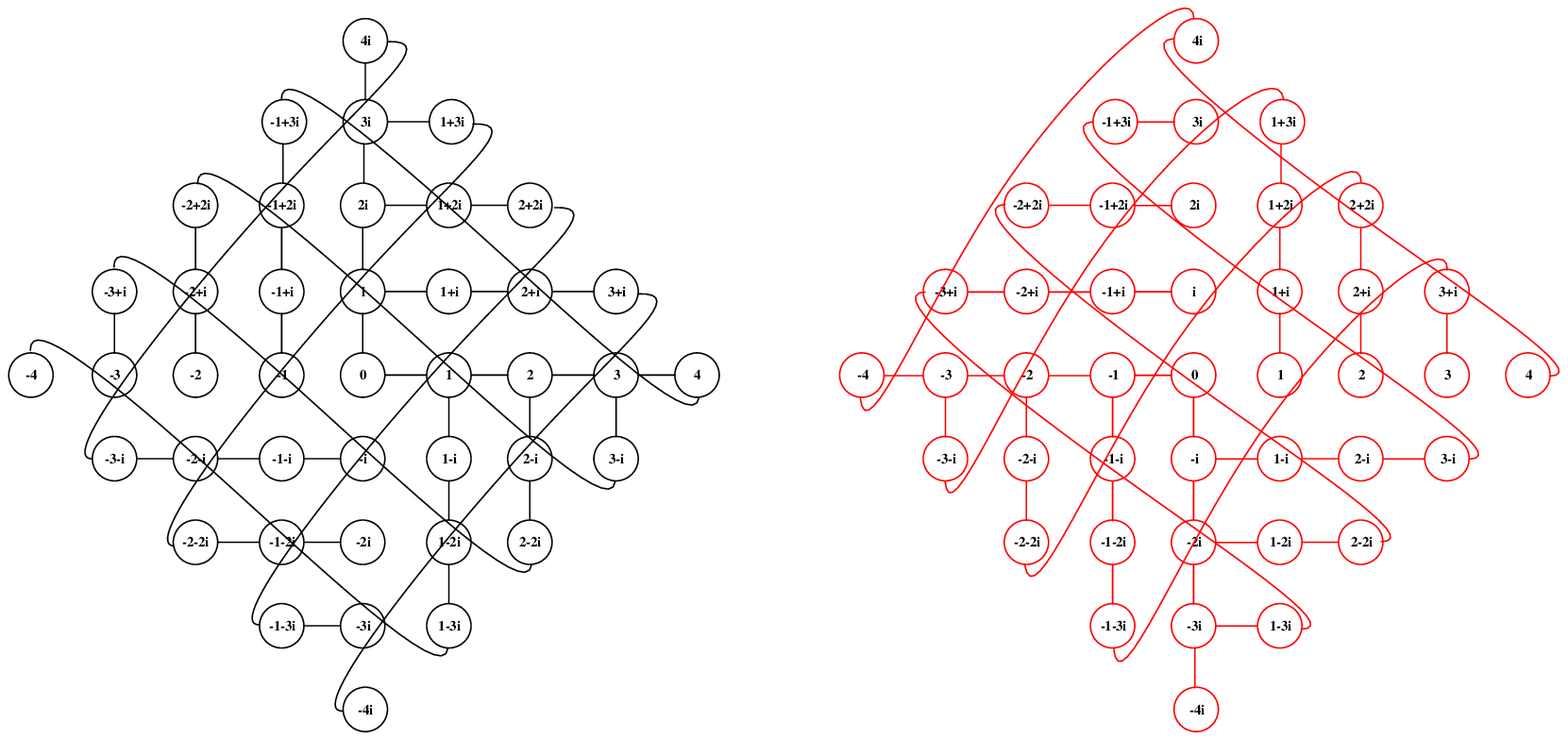} \caption{The trees $B_{4}$ and $R_{4}$}%
\label{45BlackRedST}%
\end{figure}

\bigskip
\begin{lemma}
\label{Connected} Each of the graphs $B_{k}$ and $R_{k}$ is a connected graph.
\end{lemma}

\begin{proof}
First, we will show that $B_{k}$ is connected. $B$ is connected by
construction. $\rho(B)$ is a rotation of B, and hence, it is connected.
$B \cup \rho(B)$ is connected since $B\cap\rho(B)=\left\{  0\right\}  \neq\phi$.
There exist paths between:
every node in $A$ and some node in $\rho(B)\subset A$,
every node in $\rho^{3}(A)$ and some node in $B\subset\rho^{3}(A)$,
every node in $\rho(A)$ and some node in $\rho^{3}(A)$ through $\rho(W^{B})$,
and every node in $\rho^{2}(A)$ and some node in $A$ through $W^{B}$.
Thus, $B_{k}$ is connected.
A similar argument can be used to show that $R_k$ is connected.
\end{proof}

\bigskip
\begin{corollary}
\label{EDST} $B_{k}$ and $R_{k}$ are edge-disjoint spanning trees in $G_{k}$.
\end{corollary}

\begin{proof}
Following Lemma~\ref{Connected} and Proposition~\ref{BasicRB}, $B_{k}$ and
$R_{k}$ are connected, edge-disjoint, and the edge size of each of them is
$\left\vert V_{k}\right\vert -1$. Therefore, they are edge-disjoint spanning
trees in $G_{k}$.
\end{proof}

\bigskip
In the remaining text we will assume that the spanning trees $B_{k}$ and
$R_{k}$ are rooted at node $0$.
The following three lemmata will be useful in proving our main result expounded
in~Theorem~\ref{BasicProp} and Theorem~\ref{EXTEND}.

Let
$H^{B}=A\cup W^{B}\cup\rho^{2}\left(  A\right)$, and
$H^{R}=\rho\sigma\left(  A\right)  \cup W^{R}\cup\rho^{3}\sigma\left(A\right)$.

\bigskip
\begin{lemma}
\label{QsPathSize}
Each path in $H^{B}$ or in $H^{R}$ is horizontal and not longer than $k$ or $(k+1)$, respectively.
%
\end{lemma}

\begin{proof}
The paths being horizontal or vertical follow from the definition of the sets $A,B$ and
$W^{B},W^{R}$. Each path in $A$ or $\rho^{2}\left(  A\right)  $ is of some
length $i$, $i\in\left\{  0,\ldots,k-1\right\} $. Any path in $W^{B}$ is of
length one. A path of length~$i$ in $A$ is connected by an edge in $W^{B}$
to a single path in $\rho^{2}\left(  A\right) $ of length $k-i-1$. The
maximum length of a path in $H^{B}$ is therefore $i+1+\left(  k-i-1\right)=k$.
%
Fig.~\ref{paths1} illustrates $H^B$ paths when $k=4$.
\begin{figure*}[t!]
    \centering
    \begin{subfigure}[b]{0.5\textwidth}
        \centering
        \includegraphics[scale=0.45]{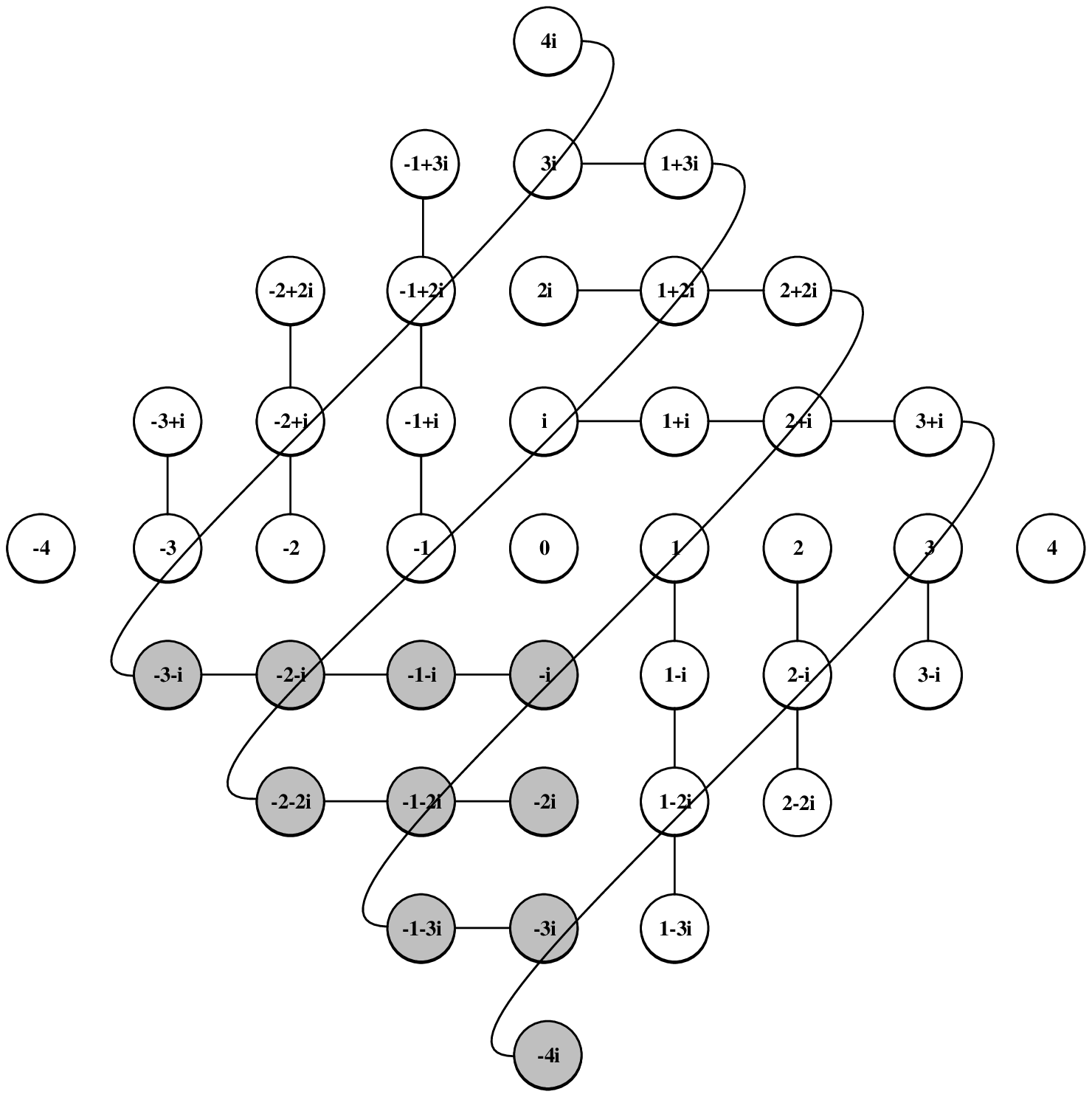}
        \caption{$W^B$ wrapped-around edges}
    \end{subfigure}%
    ~
    \begin{subfigure}[b]{0.5\textwidth}
        \centering
        \includegraphics[scale=0.45]{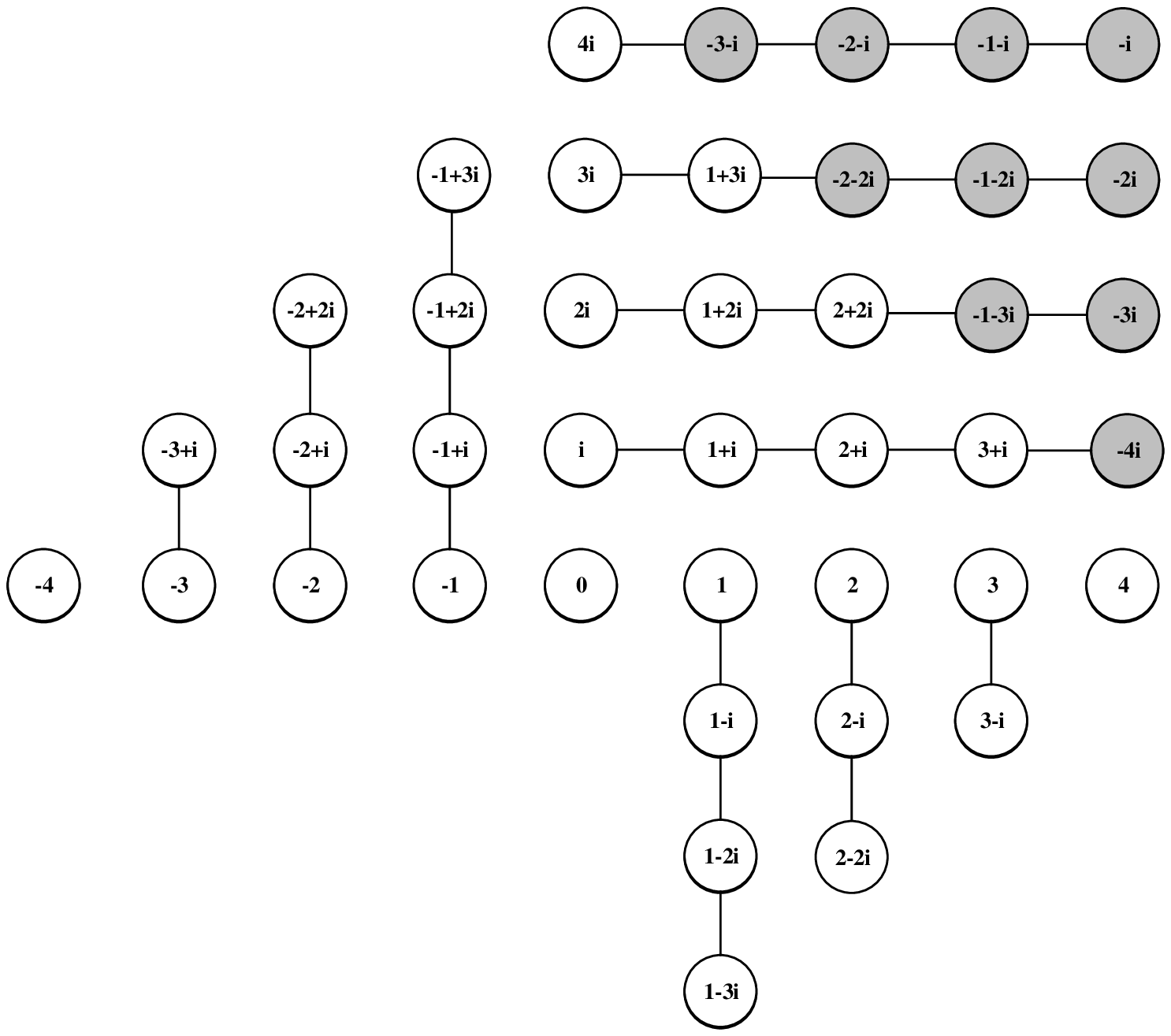}
        \caption{$W^B$ straightened edges}
    \end{subfigure}
    \caption{$H^B$ paths in $B_4$}
    \label{paths1}
\end{figure*}
Each path in $\rho\sigma\left(A\right) $ or in $\rho^{3}\sigma\left(A\right) $
is of some length $i$, $i\in\left\{  0,\ldots,k-1\right\}$.
Any path in $W^{R}$ is of length one.  A path of length $i$ in $\rho\sigma\left(A\right)$
is connected by an edge in $W^{R}$ to a single path in
$\rho^{3}\sigma\left(  A\right)$ of length $k-i$. The maximum length of a path in
$H^{B}$ is therefore $i+1+\left(k-i\right)  =k+1$.
Fig.~\ref{paths2} illustrates $H^R$ paths when $k=4$.
\begin{figure*}[t!]
    \centering
    \begin{subfigure}[b]{0.5\textwidth}
        \centering
        \includegraphics[scale=0.45]{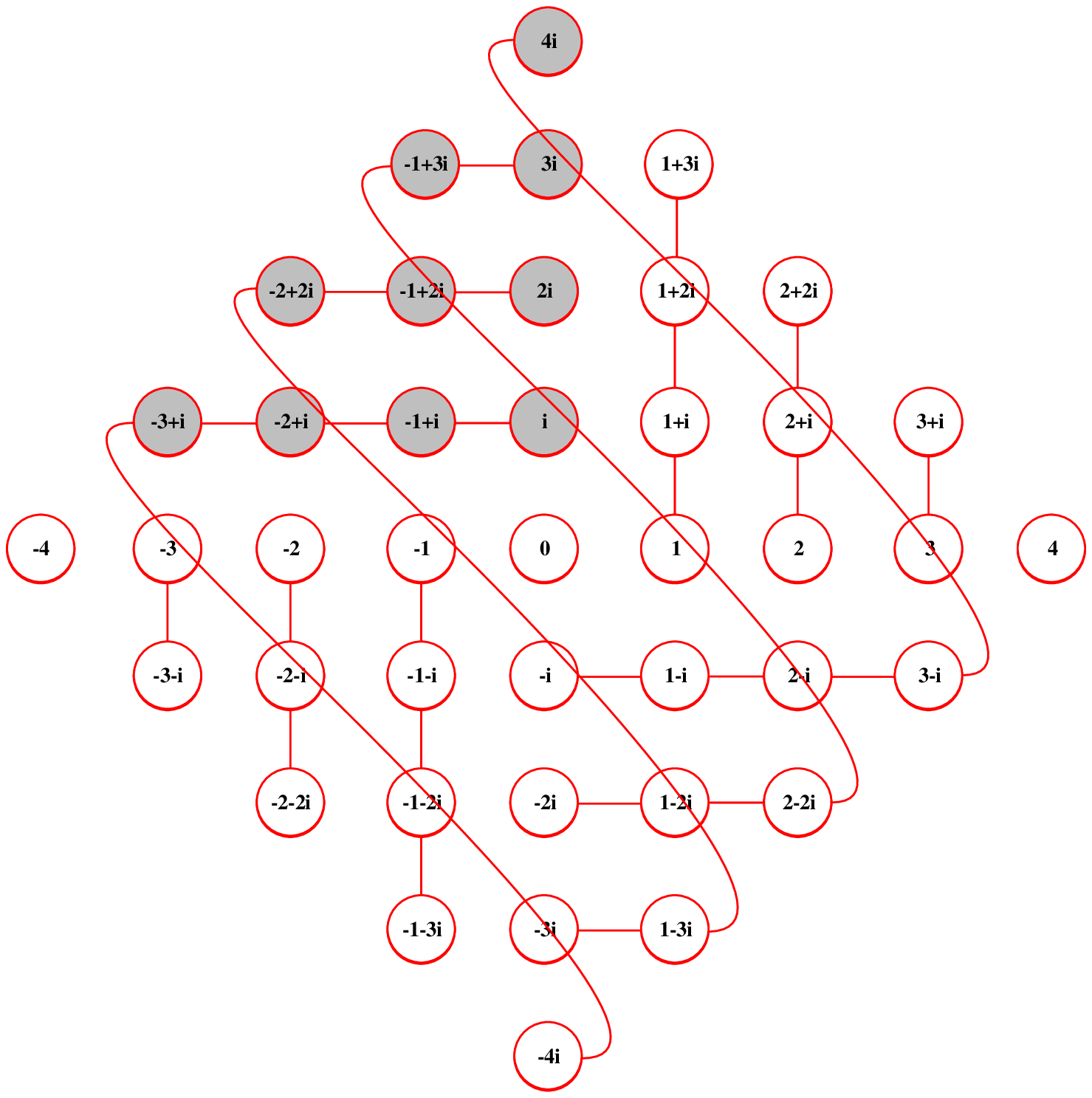}
        \caption{$W^R$ wrapped-araound edges}
    \end{subfigure}%
    ~
    \begin{subfigure}[b]{0.5\textwidth}
        \centering
        \includegraphics[scale=0.45]{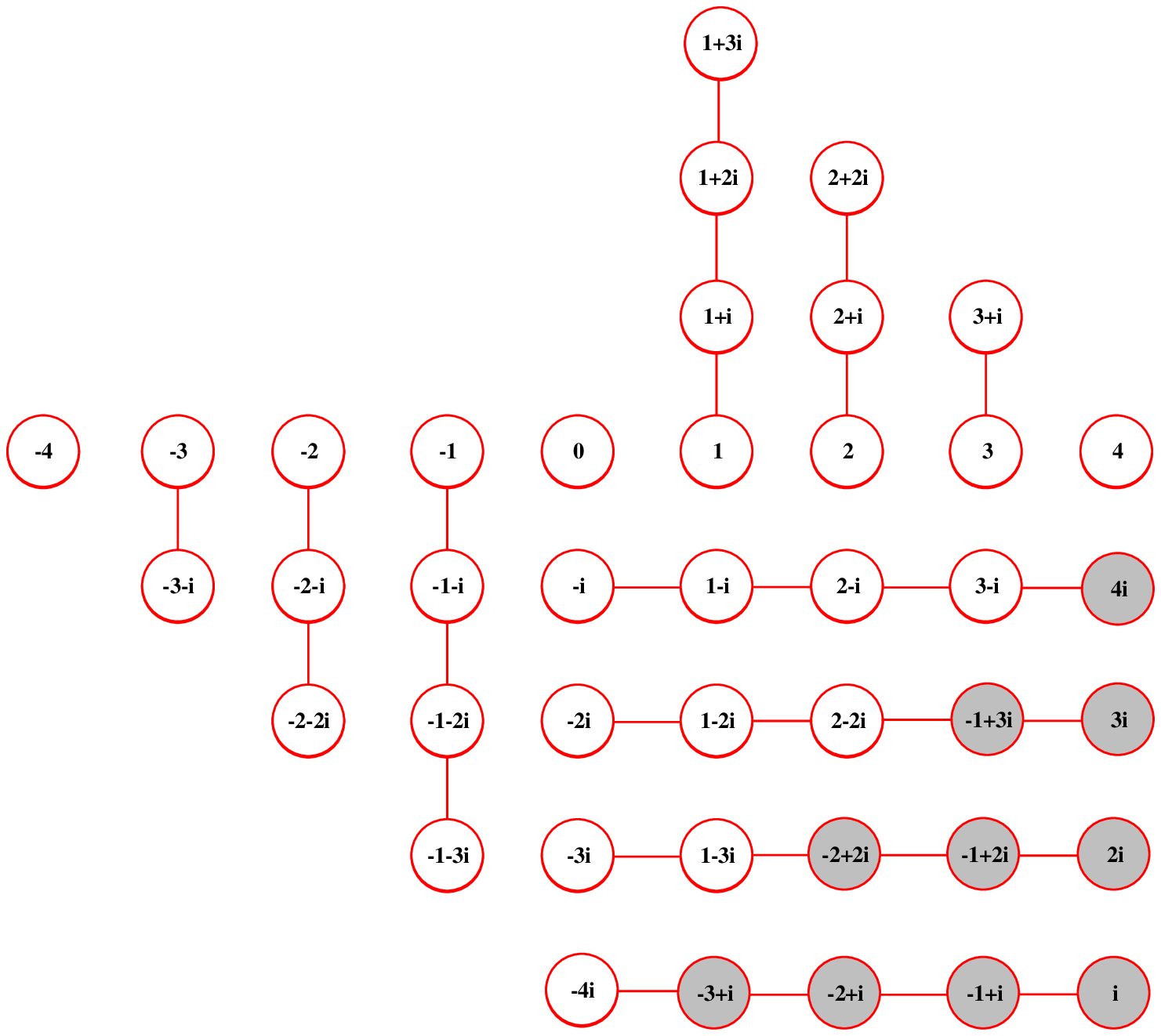}
        \caption{$W^R$ straightened edges}
    \end{subfigure}
    \caption{$H^R$ paths in $R_4$}
    \label{paths2}
\end{figure*}
\end{proof}

\bigskip
\begin{lemma}
\label{Depth}
The longest path in $B_{k}$ or in $R_k$
starting at node $0$ is of length $2k$, $k \ge 2$.
\end{lemma}

\begin{proof}
It is useful to observe that the properties of the mapping $\rho$ and
Lemma~\ref{QsPathSize} imply that the length of a path is at most~$k$ in
$\rho\left(H^{B}\right) $ and at most $(k+1)$ in $\rho\left(  H^{R}\right)$.
\newline
The longest path in $\rho\left(B\right)$ is of length $k$. Any
path in $B_{k}$ starting at node $0$ having an initial segment in
$\rho\left(B\right)$ may continue through $H^{B}$ only.
Following Lemma~\ref{QsPathSize}, the length of such path is at most $k+k=2k$.
This may happen only if the initial segment is of length $k$ and,
consequently, the path leads to node $-\mathbf{i}$.
\newline
Similar reasoning applies to $B$ and $\rho\left(H^{B}\right)$.
Any path in $B_{k}$ starting at node $0$ having an initial segment in $B$
may continue through $\rho\left(H^{B}\right)$ only.
The longest path of this kind is, therefore, of length $2k$ and leads to node $-1$.
\newline
The longest path in $\sigma\left(B\right)$ starting at node $0$ is of
length $k$ and cannot be extended. Any other path in $\sigma\left(B\right)$
is of length $k-1$ or less. Any path in $R_{k}$ starting at node $0$
having an initial segment in $\sigma\left(B\right)$ may continue through
$H^{R}$ only. Following Lemma~\ref{QsPathSize}, the length of such path is at most
$\left(  k-1\right)  +\left(  k+1\right)  =2k$. This may happen only if the
initial segment is of length $(k-1)$, and thus, the path leads to
node $\mathbf{i}$.
\newline
The longest path in $\sigma\rho\left(B\right)$
starting at node $0$ is of length $k$ and may continue in $R_{k}$ by a path
of length two or less. The total length is then not larger than $k+2 \leq 2k$,
since $k \geq 2$; the equality takes place for a path leading to node $k$ only if
$k=2$. Any path in $R_{k}$ starting at node $0$ having an initial
segment shorter than $k$ in $\sigma\rho\left(B\right)$ may continue
through $\rho\left(H^{R}\right)$ only. Accordingly, The longest path of this kind is
of length $\left(k-1\right)  +\left(  k+1\right)  =2k$. This may
happen only if the initial segment is of length $k-1$, and hence, the
path leads to node $1$.
\end{proof}

\bigskip
\begin{lemma}
\label{HorVert}
A horizontal path and a vertical path in $G_k$, each being of length
$k+1$ or less, can have at most two common nodes. This may happen only if one of
the paths is of length $k+1$ and the two common nodes are its starting and ending nodes.
\end{lemma}

\begin{proof}
Assume a horizontal and a vertical paths, each being of length $k+1$ or less,
having at least one node in common. Since $G_{k}$ is vertex-transitive, we
can assume, without loss of generality, that the common node is $0$.
The paths initial segments of size $k$ starting from $0$ are line segments and cannot intersect
in more than one node. Therefore, another
common node may exist only if one of the paths is of length $k+1$
and $0$ is its starting node, the last edge of such path is a wrap-around edge,
and the other common node is the ending node of the path.
\end{proof}

\bigskip
\begin{theorem}
\label{BasicProp} $B_{k}$ and $R_{k}$ are edge-disjoint node-independent
spanning trees in $G_{k}$ each of depth $2k$, $k \geq 2$.
\end{theorem}

\begin{proof}
Let $k\geq2$. Then $B_{k}$ and $R_{k}$ are edge-disjoint spanning trees in
$G_{k}$ by Corollary~\ref{EDST}.
Since both trees are rooted in $0$ and following Lemma~\ref{Depth},
the depth of each tree is $2k$.
To complete the proof of this theorem, we
need to prove that $B_{k}$ and $R_{k}$ are node independent. Let $v\in V_{k}$.
Assume, in contrary, a node $u$ being an intermediate node of both: the path from
$0$ to $v$ in $B_{k}$, and the path from $0$ to $v$ in $R_{k}$. Then $u$ is of
degree two in both trees. The node $v$ cannot belong to $V\left(X\right)$,
where $X=B\cup\rho(B)\cup\sigma(B)\cup\sigma\rho(B)$, since otherwise in one
of the trees the path to $v$ leads exclusively through nodes of degree three.
Let  $v\in V\left(  G_{k}\right)  -V(X)$. The path from $0$ to $v$ starts in
both trees by a segment in $X$ and then continues by a segment in one of
$H^{B},H^{R},\rho\left(  H^{B}\right)  ,\rho\left(  H^{R}\right)  $. The node
$u$ cannot belong to $V\left(  X\right)  $, since the only such $u$ being of
degree two in both trees is $u=k\mathbf{i}$; the only path in $R_{k}$ with
intermediate node $k\mathbf{i}$ leads to node $k$, while the path to $k$ in
$B_{k}$ does not contain $k\mathbf{i}$. Hence, $u$ belongs to one of
$H^{B}\cap\rho\left(  H^{R}\right)  $ or $\rho\left(  H^{B}\right)  \cap H^{R}$.
It follows by Lemma~\ref{QsPathSize} and the~$\rho$
mapping that $u$ belongs to an intersection of a
horizontal and a vertical paths, each being of length $k+1$ or less.
Lemma~\ref{HorVert} implies that one of these paths is of length~$(k+1)$ and
its starting and ending nodes are~$u$~and~$v$.
This leads to a contradiction since such a path may exist only
if~$u,v \in V(X)$.
%
\end{proof}

\bigskip
\begin{remark}
The only two edges in $E_{k}$ not belonging to the graphs $B_{k}$ or $R_{k}$ are
the edges $\left(  -k,-k\mathbf{i}\right)  $ and $\left(  k,-k\mathbf{i}\right) $.
\end{remark}

\bigskip
$B_1$ and $R_1$ are depicted in Fig.~\ref{b1r1}. As it can be seen, $B_1$ is of depth two;
however, $R_1$ is of depth three.
We can define $R'_k = (V(R_k), E(R_k) - \{(k, i)\} + \{(k, -i)\})$.
$R'_1$ is illustrated in Fig.~\ref{r1prime}, and its depth is clearly two.
%
%
\begin{figure*}[t!]
    \centering
    \begin{subfigure}[b]{0.5\textwidth}
        \centering
        \includegraphics[scale=0.45]{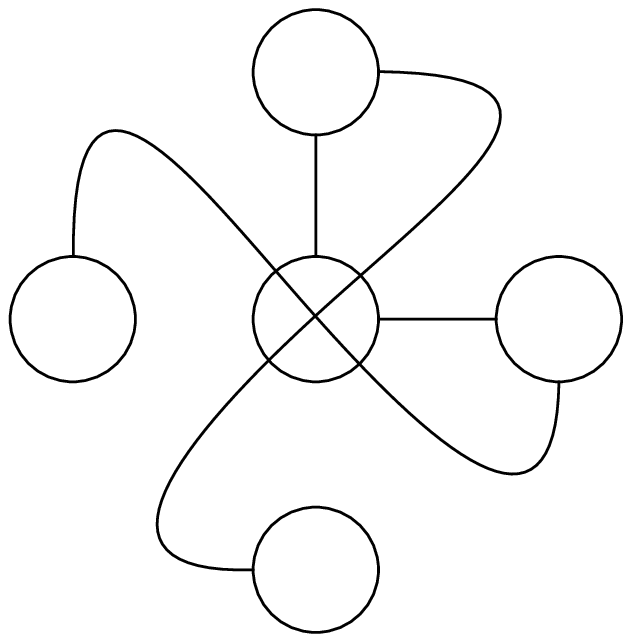}
        \caption{$B_1$}
        \label{b1}
    \end{subfigure}%
    ~
    \begin{subfigure}[b]{0.5\textwidth}
        \centering
        \includegraphics[scale=0.45]{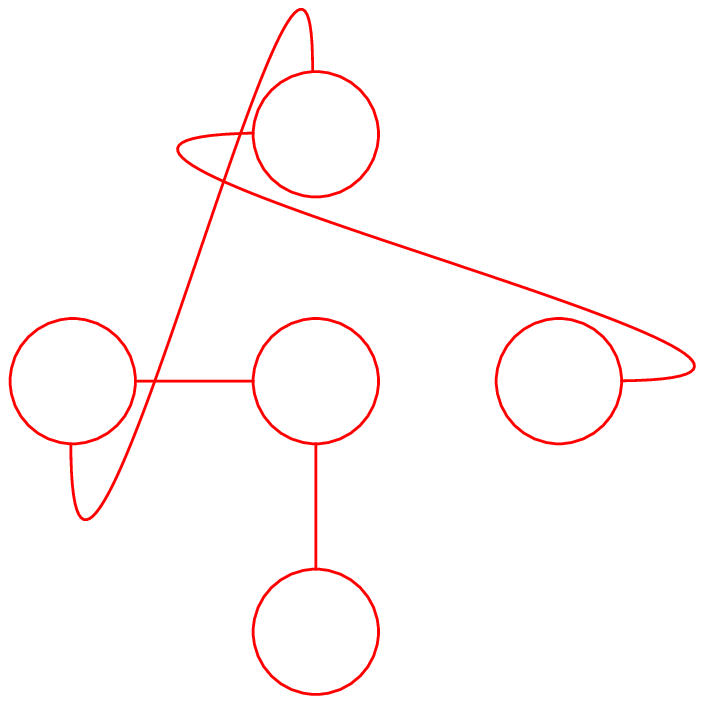}
        \caption{$R_1$}
        \label{r1}
    \end{subfigure}
    ~
    \begin{subfigure}[b]{0.5\textwidth}
        \centering
        \includegraphics[scale=0.45]{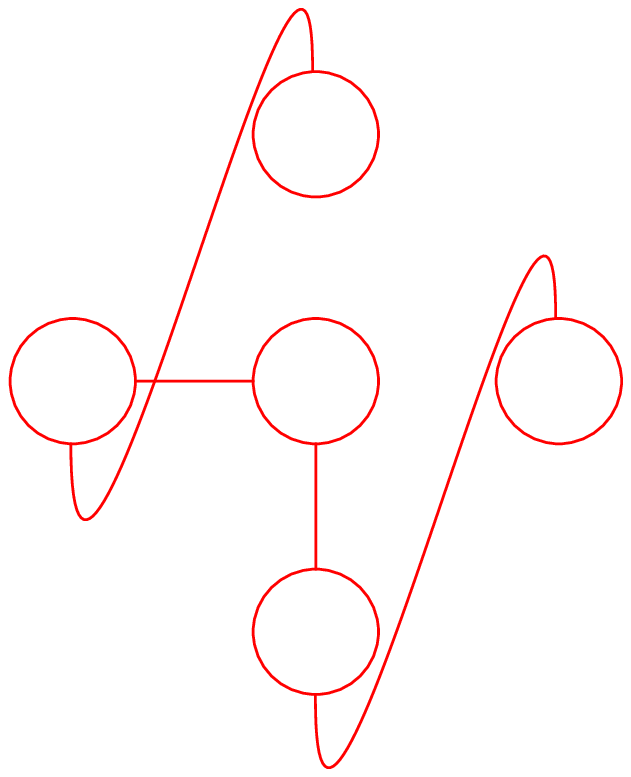}
        \caption{$R'_1$}
        \label{r1prime}
    \end{subfigure}
    \caption{$B_1$, $R_1$, and $R'_1$}
    \label{b1r1}
\end{figure*}

\bigskip
\begin{remark}
\label{BkR'k}
The path $P_{R'_k}(0, k)$ is shorter than $P_{R_k}(0, k)$ by one, and
the only two edges in $E_{k}$ not belonging to the graphs $B_{k}$ or $R'_{k}$ are
$\left(-k,-k\mathbf{i}\right)  $ and $\left(  k,k\mathbf{i}\right)$.
\end{remark}

\bigskip
Theorem~\ref{EXTEND} below extends Theorem~\ref{BasicProp} to prove that $B_k$ and $R'_k$ are edge-disjoint
node-independent spanning trees in $G_k$, $k \ge 1$.

\bigskip
\begin{theorem}
\label{EXTEND}
$B_k$ and $R'_k$ are edge-disjoint
node-independent spanning trees in $G_k$ each of depth $2k$, $k \ge 1$.
\end{theorem}

\begin{proof}
By Theorem~\ref{BasicProp}, $B_k$ and $R_k$ are edge-disjoint node-independent spanning trees in $G_k$, $k \ge 2$.
Disconnecting a leaf node from $R_k$ and reconnecting it to another node results again in a connected graph.
This implies that $R'_k$ is a spanning tree in $G_k$.
\newline
Only one edge in $R'_k$ does not belong to $R_k$, and this edge is unused in $B_k$.
Thus, $B_k$ and $R'_k$ are edge-disjoint spanning trees in $G_k$.
\newline
All paths in $R'_k$ are identical to those in $R_k$ except for the path that leads to $k$.
To prove that $B_k$ and $R'_k$ are node-independent spanning trees, we only need
to verify that the paths lead to $k$ in $B_k$ and in $R'_k$ are node-disjoint.
In $B_k$, the path form $0$ to $k$ leads exclusively through nodes of degree three. These nodes must be
leaves in $R'_k$ and cannot be intermediate nodes in any path.
Hence, $B_k$ and $R'_k$ are edge-disjoint node-independent spanning trees in $G_k$, $k \ge 2$.
\newline
The distance from $0$ to $k$ in $R'_k$ equals the distance from $0$ to $-k \mathbf{i}$ plus one = $k+1 < 2k$, $k \ge 2$.
Therefore, $B_k$ and $R'_k$ are of depth $2k$, $k \ge 2$.
\newline
The theorem can be easily verified for $k=1$, and this concludes the theorem proof.
\end{proof}

\bigskip
$B_4$ and $R'_4$ are depicted in Fig.~\ref{BT4}.

\begin{figure}[H]
\centering
\includegraphics[scale=0.95]{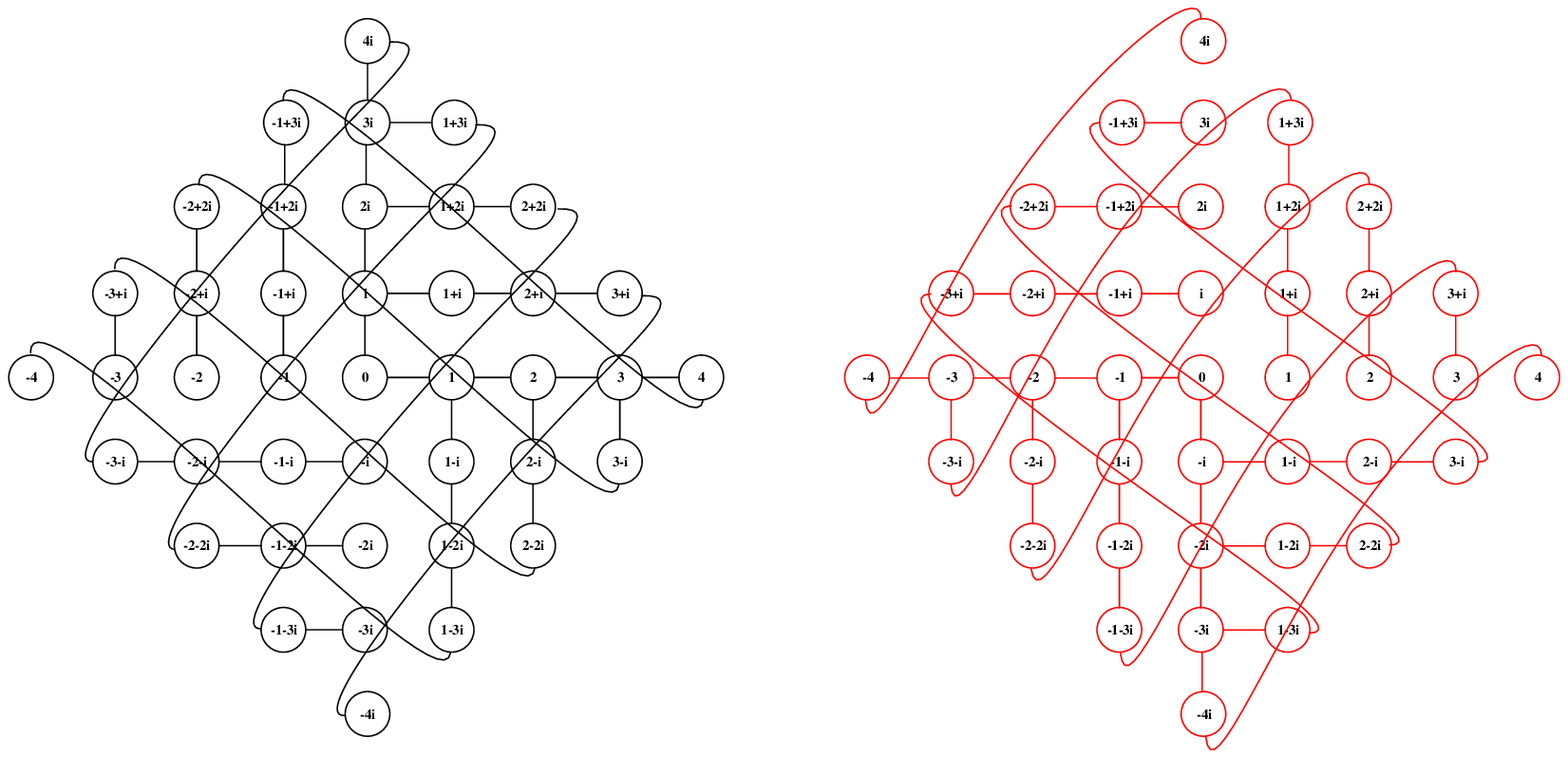} \caption{The trees $B_{4}$ and $R'_{4}$}%
\label{BT4}%
\end{figure}

\section{Routing in $G_{k}$ using edge-disjoint node-independent
trees\label{SectionRouting}}

The existence of two edge-disjoint node-independent spanning trees in $G_{k}$
makes it possible to tolerate either one node or one edge failure. Sending a
message along both trees guarantees the message delivery in case of a single
failure of either a node or an edge. The trees also can be utilized to
securely distribute a message by splitting it into two packages each of which
is sent to the destination along a different tree. Since the trees are node
independents, the paths to the destinations in those trees are disjoint. This
guarantees that only the destination will receive all the packages. Note that
in both of the above applications the communications must be initiated from
the root. However, it is possible for each node~$v$ to implicitly construct
two edge-disjoint node-independent spanning trees rooted in $v$ itself.

We present routing algorithms for such usage when the source node
$s_{1}+s_{2}\mathbf{i}$ sends a message to a destination node $d_{1}+d_{2}\mathbf{i}$
along $B_{k}$ or $R'_{k}$, $k \ge 1$. The algorithms are based on the vertex-transitive property of $G_k$.
The source node is mapped to node~$0$. The destination node and every transient
node are mapped accordingly.
Algorithm~1 outlines the routing in the source node.
The tree type and the destination mapped location determine the message direction
from the source node.
Table~\ref{SourceRoutBlackRed} summarizes all possible directions form the source node and the exclusive condition
associated with each direction.

Algorithm~2 describes the routing in a non-source node. If a receiving node is the destination,
then no further routing is required.
Otherwise the receiving node is a transient node and the message is rerouted after deciding its new direction.
By default, a message keeps flowing in the same direction unless it needs to make a turn.
A message makes a turn only if the transient node lays on the horizontal axis or the vertical axis,
and the destination lays on a continuing path in~$H^B$ or~$H^R$.
Table~\ref{TransRoutBlackRed} shows the two exclusive possible conditions in which a message needs to make a turn,
and it associates each condition with the message new direction. If non of these conditions holds, then the default case
applies.

\begin{algorithm}[H]
\caption{Routing algorithm for a source node $s_{1}+s_{2}\mathbf{i}$ to send a message
to a destination node $d_{1}+d_{2}\mathbf{i}$ through $B_k$ or $R'_k$}
\begin{algorithmic}[1]
\STATE Compute the mapped location of the destination:\\
\hskip\algorithmicindent \hskip\algorithmicindent
$d_{1}^{m}+d_{2}^{m}\mathbf{i=}\left(  \left(  d_{1}+d_{2}\mathbf{i} \right)
-\left(  s_{1}+s_{2}\mathbf{i}\right)  \right)  \operatorname{mod} \alpha_{k}$
\STATE Use Table \ref{SourceRoutBlackRed} to determine the message direction  \textbf{Dir}
\STATE Create the message header: $\left(  \left\langle s_{1},s_{2}\right\rangle ,\left\langle d_{1}^{m},d_{2}^{m}\right\rangle, \textbf{Dir} \right)$
and send the message
\end{algorithmic}
\end{algorithm}

\begin{algorithm}[H]
\caption{Routing algorithm for a non-source node $t_{1}+t_{2}\mathbf{i}$}
\begin{algorithmic}[1]
\STATE Read the message header $\left(  \left\langle s_{1},s_{2}\right\rangle ,\left\langle d_{1}^{m},d_{2}^{m}\right\rangle, \textbf{Dir} \right)$
\STATE Compute the current node mapped location:\\
\hskip\algorithmicindent \hskip\algorithmicindent $t_{1}^{m}+t_{2}^{m}\mathbf{i=}\left(  \left(  t_{1}+t_{2}\mathbf{i}\right)  -\left(  s_{1}+s_{2}\mathbf{i}\right)  \right)\operatorname{mod}\alpha_{k}$
\IF{$\left\langle d_{1}^{m},d_{2}^{m}\right\rangle = \left\langle t_{1}^{m},t_{2}^{m}\right\rangle $}
\STATE Accept the message and \textbf{exit}
\ELSE
\STATE Use Table \ref{TransRoutBlackRed} to decide the message direction \textbf{NewDir}
\STATE Create the new message header: $\left(  \left\langle s_{1},s_{2}\right\rangle ,\left\langle d_{1}^{m},d_{2}^{m}\right\rangle, \textbf{NewDir} \right)$ and send the message
\ENDIF
\end{algorithmic}
\end{algorithm}

\begin{table}
\centering
{\footnotesize
\begin{tabular}
[c]{|l|l|l|l|}\hline
Condition & $B_k$ Direction & $R'_k$ Direction & $d_{1}^{m}+d_{2}^{m}\mathbf{i}$ Location\\\hline
$d_{1}^{m}\times d_{2}^{m}>0$ & \multicolumn{1}{|c|}{$+\mathbf{i}$} &
\multicolumn{1}{|c|}{$-1$} & Quadrant 1 or 3\\\hline
$d_{1}^{m}\times d_{2}^{m}<0$ & \multicolumn{1}{|c|}{$+1$} &
\multicolumn{1}{|c|}{$-\mathbf{i}$} & Quadrant 2 or 4\\\hline
$d_{2}^{m}=0$ & \multicolumn{1}{|c|}{$+1$} & \multicolumn{1}{|c|}{$-1$} &
Horizontal Axis\\\hline
$d_{1}^{m}=0$ & \multicolumn{1}{|c|}{$+\mathbf{i}$} &
\multicolumn{1}{|c|}{$-\mathbf{i}$} & Vertical Axis\\\hline
\end{tabular}
}%
\caption{The source node routing table}
\label{SourceRoutBlackRed}
\end{table}

\begin{table}
\centering
{\footnotesize%
\begin{tabular}
[c]{|l|l|c|cc|}\hline
Condition & \textbf{NewDir} & \multicolumn{3}{|c|}{Location}\\
\cline{3-5}
&  & $t_{1}^{m}+t_{2}^{m}\mathbf{i}$ & \multicolumn{2}{|c|}{$d_{1}^{m}+d_{2}^{m}\mathbf{i}$}\\
\cline{4-5}
&  &  & $B_k$ & \multicolumn{1}{|c|}{$R'_k$} \\
\hline
$(  t_{2}^{m}=0)$  AND &  & On Horizontal  Axis & Quadrant 2 or
4 & \multicolumn{1}{|c|}{Quadrant 1 or 3}\\
$(  t_{1}^{m}=d_{1}^{m}$ OR $t_{1}^{m}=d_{1}^{m}%
+k+1$ OR $t_{1}^{m}=d_{1}^{m}-k)  $ &
\multicolumn{1}{|c|}{$-\mathbf{i}$} &  &  & \multicolumn{1}{|c|}{}\\\hline
$\left(  t_{1}^{m}=0\right)$ AND & \multicolumn{1}{|c|}{} & On Vertical
Axis & Quadrant 1 or 3 & \multicolumn{1}{|c|}{Quadrant 2 or 4}\\
$(  t_{2}^{m}=d_{2}^{m}$ OR $t_{2}^{m}=d_{2}^{m}+k+1$ OR
$t_{2}^{m}=d_{2}^{m}-k)  $ & \multicolumn{1}{|c|}{$+1$} &  &  &
\multicolumn{1}{|c|}{}\\\hline
Default & \multicolumn{1}{|c|}{\textbf{Dir}} & \multicolumn{3}{|c|}{
Keep the message flowing in the same direction.}\\\hline
\end{tabular}
}%
\caption{A transient node routing table}
\label{TransRoutBlackRed}
\end{table}

\newpage
\section{Conclusions\label{SectionConcl}\bigskip}

We introduced two constructions of edge-disjoint node-independent spanning trees in dense Gaussian networks.
By taking advantage of the node-transitivity in dense Gaussian networks,
we defined a limited number of subgraphs and deployed a rotation technique to construct the first pair of trees.
The depth of each tree in the first construction is $2k$, $k \ge 2$, where $k$ is the network diameter.
We extended the first construction to construct the second pair of trees.
The depth of each tree in the second construction is $2k$, $k \ge 1$.
Based on the second construction, we designed algorithms that can be used in fault-tolerant routing or
secure message distribution.
The source node in these algorithms is not restricted to a specific node; it could be any node in~$G_k$.

In our future work we intend to investigate constructing independent spanning trees and completely independent
spanning trees in dense Gaussian networks. Our initial investigations indicate that applying similar techniques
to those deployed in this paper could lead to fruitful outcomes.
\bibliographystyle{IEEEtranS}
\bibliography{IEEEabrv,EDNIST}





\end{document}